\documentclass[12pt]{article}

\usepackage{amsmath,amsthm,amsfonts,amssymb}
\usepackage{xcolor}
\usepackage{graphicx}
\usepackage{xfrac}
\usepackage{caption}
\usepackage{subcaption}
\usepackage{authblk}

\newcommand{\J}{{\ensuremath{{\mathrm{J}}}}}
\newcommand{\A}{{\ensuremath{{\mathrm{A}}}}}
\newcommand{\df}{{\ensuremath{\mathrm{df}}}}
\newcommand{\en}{{\ensuremath{\mathrm{en}}}}
\newcommand{\eq}{{^*}}
\newcommand{\nv}{{\ensuremath{\mathrm{nv}}}}
\newcommand{\vac}{{\ensuremath{\mathrm{v}}}}
\newcommand{\R}{{\ensuremath{\mathcal{R}}}}
\newcommand{\bydef}{:=}

\newtheorem{theorem}{Theorem}
\newtheorem{proposition}[theorem]{Proposition}

\newtheorem{remark}{Remark}

\newtheorem{definition}{Definition}

\definecolor{lightblue}{rgb}{.9, .9, 1}
\definecolor{lightred}{rgb}{1, .9, .9}
\definecolor{lightpurple}{rgb}{.9, .8, .9}
\definecolor{lightgray}{rgb}{.8, .8, .8}
\definecolor{gray_lightblue}{rgb}{.7, .7, .8}

\title{Social \textit{vs.} individual age-dependent costs of imperfect vaccination}

\begin{document}

\author[1]{Fabio A. C. C. Chalub\footnote{Corresponding author: facc@fct.unl.pt}}
\author[1]{Paulo Doutor}
\author[1]{Paula Patrício}
\author[1]{Maria do Céu Soares}

\affil[1]{Center for Mathematics and Applications (NOVA Math) and Department of Mathematics, NOVA FCT, Caparica, 2829-516, Portugal}

\maketitle

\begin{abstract}
In diseases with long-term immunity, vaccination is known to increase the average age at infection as a result of the decrease in the pathogen circulation. This implies that a vaccination campaign can have negative effects when a disease is more costly (financial or health-related costs) for higher ages. This work considers an age-structured population transmission model with imperfect vaccination. We aim to compare the social and individual costs of vaccination, assuming that disease costs are age-dependent, while the disease's dynamic is age-independent. A model coupling pathogen deterministic dynamics for a population consisting of juveniles and adults, assumed to be rational agents, is introduced. The parameter region for which vaccination has a positive social impact is fully characterized and the Nash equilibrium of the vaccination game is obtained. Finally, collective strategies designed to promote voluntary vaccination, without compromising social welfare, are discussed.
\end{abstract}

\newcommand{\sep}{ \& }
\textbf{Keywords:}
Vaccination games\sep epidemic models\sep Nash equilibria\sep age-structured models.

\section{Introduction}\label{sec1}

In this work, we examine highly contagious diseases that provide permanent immunity. Assuming that a vaccine is not available, individuals will be exposed to the pathogen in their early years and will acquire life-long immunity. This means that in a non-vaccinated population, the disease will infect mainly children~\cite{Panagiotopoulos1999}.  On the other hand, in a fully vaccinated population with a perfect vaccine, the disease will cease to exist. 

However, when vaccination coverage is moderate or if the vaccine confers imperfect protection, the average age of first exposure to the pathogen will be larger than in a non-vaccinated population. Depending on the precise characteristics of the disease, the population-wide effect may be worse than in both previous cases. For example,  varicella (chickenpox), rubella, mumps, and Zika have generally mild effects with a few severe cases in juveniles, but when the virus infects adults, particularly pregnant women, the health consequences can be severe~\cite{Riera-Montes,Daley,Abushouk}.

 We introduce a mathematical model where we consider imperfect vaccination. We first identify the vaccine coverage that avoids pernicious effects depending on vaccine and infection costs. Additionally, we analyze how human behavior affects vaccination compliance. We aim to compare population and individual interests towards vaccination and examine two particularly relevant issues: waning immunity from vaccines and human behavior.

In the model of the present work, the various costs associated with the disease are assumed to be age-dependent, while the disease's dynamic is age-independent. We divide the population into juveniles and adults and assume that both groups' transmission rates and average infectious periods are the same. The only modeling assumption that differs between both age groups is that both fertility and mortality rates are zero for juveniles, and non-zero otherwise. 

Infection-related costs may include different convalescence times, absence from work, increased need for medical care, and even emotional costs, such as the interruption of pregnancy. It may also include later manifestations of infections, such as shingles, an inflammation of the sensory ganglia caused by reactivation of the herpes zoster virus, the same that causes varicella, after a dormancy period that may span over decades~\cite{Kawai_BMC2014}.

In this work, we assume that individual costs are higher for infected adults than for juveniles. It is necessary to include a cost for vaccinated individuals, including not only financial costs, and possible adverse reactions, but also --- real or imaginary --- vaccine risks. Finally, we introduce a social cost function as the sum of individual costs for all individuals in a given population.

The immunity induced by vaccines is different from the natural immunity provided by the disease. For example, the varicella-zoster virus may stay dormant in the body even after the acute period and reactivate later in life causing herpes zoster (shingles). That reactivation may be prevented by the continuous immune response to the presence of the virus, which means that the absence of virus circulation, due to universal vaccination may accelerate immune waning and increase the incidence of herpes zoster among adults who previously had varicella. The varicella vaccine provides indirect protection against shingles not only because decreases the probability of having a primary varicella infection but also because the attenuated virus of the vaccine is less likely to reactivate~\cite{Yang_2020,Brisson_2002}. 

There is a substantial degree of variation in antibody levels in the individuals in the population and in each individual over time, cf.~\cite{Yang_2020} for empirical data in childhood diseases. Waning immunity has important dynamic effects on the spread of mumps and measles, cf.~\cite{Yang_2020}; see also~\cite{Damron_2020} for the case of pertussis. In particular, recent outbreaks of mumps seem to be linked to vaccine-induced waning immunity~\cite{Lewnard2018, Cohen2007}. 

In this work, we divide the vaccinees into two groups according to the induced protection: those permanently protected and those protected only during a certain period. In a sense, this is a simplified version of the measles model~\cite{Mossong1999}; see also \cite{Mossong2003} for a more detailed age-structured version. In these last two works, different vaccine-induced immune classes are introduced, considering permanent immunity, temporary immunity, and no immunity at all.

To highlight the difference in costs among the two age groups considered in this work, we will consider that for vaccinated individuals with temporary immunity, the typical wanning time will correspond to the juvenile-adult transition.

In any vaccination campaign the human willingness to be vaccinated cannot be disregarded. The central paradox, already identified in~\cite{Bauch04}, is that if a vaccination scheme is well-succeeded, the risk associated with the disease will decrease to the point that the population will start to question the need for the vaccine. Just before vaccine coverage reaches the herd immunity threshold, rational individuals will stop being vaccinated as the perceived risk of the vaccine will equal the perceived risk of the disease, which will be small at this point. Therefore, herd immunity will not be obtained through vaccination without further incentives to be vaccinated (and to vaccinate the dependents) or punishment of non-vaccinated individuals (e.g., the exclusion of the school system). Since the seminal work~\cite{Bauch04}, other models considered the coupling between the deterministic disease dynamics with game-theoretical models for individual decisions within the population, cf.~\cite{Doutor2016,Manfredi2013, WANG20161, Chang2020,Laguzet,Villota_2024}.

We finish the introduction with the outline of the paper. In Section~\ref{sec:model}, we introduce the model and present some basic results,  including the explicit expression for the basic reproduction number, and the characterization of equilibria and their stability. In the sequel, we discuss the model, analyzing first the social costs of vaccination and then, using techniques from game theory, the effects of considering voluntary vaccination and individual interests; in particular, we define Nash equilibrium within the context of the present work. In Section~\ref{sec:numerics}, we present numerical simulations based on typical values for childhood diseases to study socially cost-efficient parameters regions, Nash equilibria of the vaccination games, parameters region such that rational individuals accept or do not accept to be vaccinated, and how shared costs between individuals and the society can dramatically influence the endemic equilibria of the model. We conclude in Section~\ref{sec:conc} with a discussion of our main findind.

\section{Methods}\label{sec:model}

\subsection{The model}

We consider an age-structured population divided into two groups: juveniles and adults, both having the same susceptibility, infectiousness, and recovery rate, but only adults reproduce and die. Each individual is vaccinated at birth with probability $p\in[0,1]$. The vaccine is imperfect, with efficacy $\lambda \in [0,1]$, meaning that for a fraction $\lambda$ it confers life-long immunity, while for the remainder $1-\lambda$ the immunity only lasts during the juvenile phase ($1/\nu$ yrs). The parameter $\lambda$ characterizes not only different levels of immunization to different individuals but also the fraction of the population with an incomplete vaccination scheme. The adult phase lasts for $1/\mu$ yrs. The model diagram is represented in Fig.~\ref{fig:SIR}. The relevant set of values is presented at Table~\ref{table:parameters}, while model variables are defined at Table~\ref{table:variables}. 

\begin{figure}%
\centering
\includegraphics[width=0.3\textwidth]{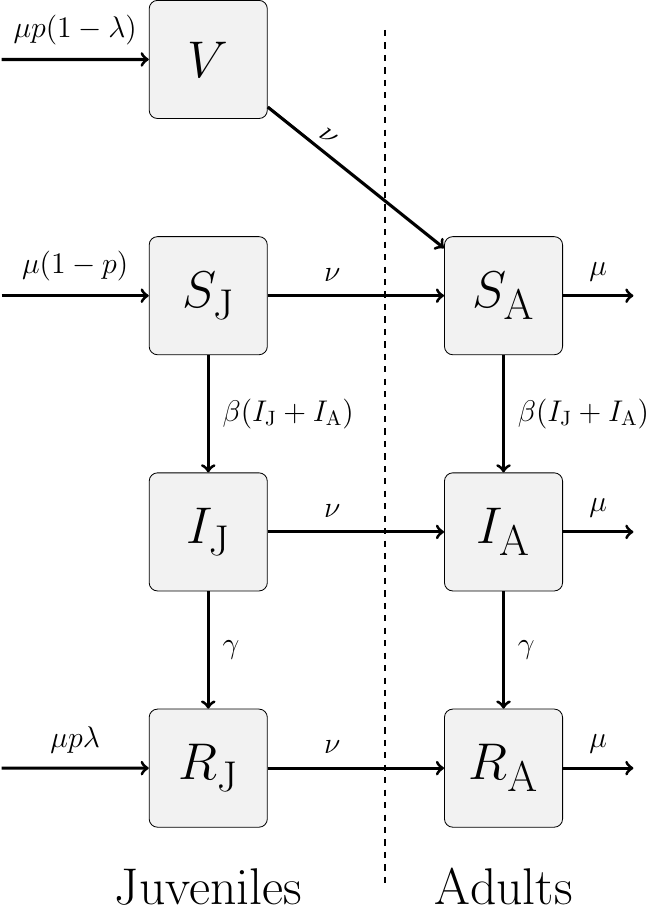}
\caption{Schematic diagram of the SIR model for juveniles and adults. Parameter $\mu$ corresponds to adult mortality and birth rates. The transition rate between both age groups is given by $\nu$. Vaccination (with coverage $p$) provides long-term immunity for a fraction $\lambda$ of the individuals and temporary (i.e., during the juvenile phase) for a fraction $1-\lambda$. Disease transmission $\beta$ and recovering $\gamma$ are assumed to be independent of the age group.}\label{fig:SIR}
\end{figure}

\begin{table}
\begin{center}
\begin{tabular}{clcc}
\textbf{Parameter} & \textbf{Description} &\textbf{Value}& \textbf{Unity} \\
\hline\hline
 $\mu>0$ &  birth/adult mortality rate &1/70&yrs$^{-1}$\\
    \hline
   $\gamma>0$ & recovering rate &365/12  &yrs$^{-1}$\\
    \hline
   $\beta>0$ & transmission rate & such that $\R_0=8$ &yrs$^{-1}$ per capita\\
    \hline
   $\nu>0$ & rate of immunity loss &1/15 &yrs$^{-1}$\\
    \hline
   $p\in [0,1]$ & vaccine coverage & &non-dimensional\\
    \hline
   $\lambda\in [0,1]$ & vaccine efficacy& &non-dimensional\\
\end{tabular}
\caption{Values used in this work. Parameters $\mu$, $\nu$, $\gamma$, and $\beta$ are not disease-specific and were chosen as an illustration in the range of varicella and rubella that served as motivation~\cite{Fine}. The value of $\beta$ was obtained from Eq.~(\ref{R0}) at demographic equilibrium.
 In Fig.~\ref{Fig4} we consider a range of values~$\R_0$.}\label{table:parameters}
\end{center}
\end{table}

\begin{table}
\begin{center}
\begin{tabular}{cp{0.8\textwidth}}
\textbf{Variable} & \textbf{Description}  \\
\hline\hline
 $V$&Fraction of individuals vaccinated at birth\\
 \hline
  $S_\J$&Fraction of susceptible juveniles\\
 \hline
 $I_\J$&Fraction of infectious juveniles,\\
 \hline
 $R_\J$&Fraction of juveniles with life-long immunity (due to recovery or vaccination)\\
 \hline
 $S_\A$&Fraction of susceptible adults\\
 \hline
 $I_\A$&Fraction of infectious adults\\
 \hline
 $R_\A$&Fraction of adults with life-long immunity (due to recovery or vaccination)\\
 \hline
 $N_\J$&Fraction of juveniles (equal to $V+S_\J+I_\J+R_\J$)\\
 \hline
 $N_\A$&Fraction of adults (equal to $S_\A+I_\A+R_\A$)
\end{tabular}
\caption{Compartment variables used in the model; c.f. Eqs.~(\ref{model:eq_V})--(\ref{model:eq_RA}). }\label{table:variables}\end{center}
\end{table}

The model is given by the following system of differential equations:

\begin{align}
\label{model:eq_V}
 V'&=\mu p(1-\lambda)N_\A -\nu V\ ,\\
\label{model:eq_SJ}
 S_\J'&=\mu(1-p)N_A-\nu S_\J-\beta(I_\J+I_\A)S_\J\ ,\\
\label{model:eq_IJ}
 I_\J'&=\beta(I_\J+I_\A)S_\J-\nu I_\J-\gamma I_\J\ ,\\
\label{model:eq_RJ}
 R_\J'&=\mu p\lambda N_\A +\gamma I_\J-\nu R_\J\ ,\\
\label{model:eq_SA}
 S_\A'&=\nu (V+S_\J)-\mu S_\A-\beta(I_\J+I_\A)S_\A\ ,\\
\label{model:eq_IA}
 I_\A'&=\beta(I_\J+I_\A)S_\A+\nu I_\J-\mu I_\A-\gamma I_\A\ ,\\
\label{model:eq_RA}
 R_\A'&=\nu R_\J+\gamma I_\A-\mu R_\A\ .
\end{align}
The total population $N=V+S_\J+I_\J+R_\J+S_\A+I_\A+R_\A$ is constant, and, therefore, we set $N(t)=1$ for all $t\ge 0$. Furthermore, we define the juvenile and adult population by $N_\J\bydef V+S_\J+I_\J+R_\J$ and $N_\A\bydef S_\A+I_\A+R_\A=1-N_\J$, respectively. Adding Eqs.~(\ref{model:eq_SA}), (\ref{model:eq_IA}) and (\ref{model:eq_RA}), we conclude that $N_\A'=\nu (1-N_\A)-\mu N_\A$. We say that a population is in \emph{demographic equilibrium} if $N_\J$ and $N_\A$ are constants. In that case
\begin{equation}\label{eq:demographic}
N_\J(t)= N_\J^*\bydef\frac{\mu}{\mu+\nu}\ ,\quad N_\A(t)=N_\A^*\bydef\frac{\nu}{\nu+\mu}\ .
\end{equation}

Both the disease-free and endemic equilibrium can be readily obtained. Their stability depends on the value of the critical parameter $\mathcal{R}_p$, obtained using the next generation matrix approach \cite{Driessche2002}.
More explicitly, we state:

\begin{theorem}\label{Equilibriastability}
 For any value of $p\in[0,1]$, there is one equilibrium solution of Eqs.~(\ref{model:eq_V})--(\ref{model:eq_RA}), called \emph{the disease-free solution}, given by
 \begin{align*}
V^{\df}&\bydef N_\J^* p(1-\lambda), &&S_\J^\df \bydef N_\J^* (1-p) ,\\
R_\J^\df&\bydef N_\J^*p\lambda, && S_\A^\df\bydef N_\A^*(1-p\lambda) , \\
R_\A^\df&\bydef N_\A^* p\lambda, && I_\J^\df\bydef I_\A^\df=0.
\end{align*}
Let the effective reproduction number be
\begin{equation}\label{Rp}
   \mathcal{R}_p\bydef\!\frac{\beta}{\gamma+\mu}\!\left[\frac{\mu+\nu+\gamma}{\nu+\gamma}S_J^\df+S_A^\df\right]
   =\frac{\beta}{\gamma+\mu}\frac{\mu(\mu+\gamma+\nu)(1-p)+\nu(\nu+\gamma)(1-\lambda p)}{(\gamma+\nu)(\mu+\nu)}\ .
\end{equation}
Then
\begin{itemize}
    \item 
If $\mathcal{R}_p\le1$ the only equilibrium solution of Eqs.~(\ref{model:eq_V})--(\ref{model:eq_RA}) is the disease-free solution, which is locally asymptotically stable if $\mathcal{R}_p<1$.

\item If $\mathcal{R}_p>1$ the disease-free solution is unstable. Furthermore, there is a second equilibrium solution of Eqs.~(\ref{model:eq_V})--(\ref{model:eq_RA}), called \emph{the endemic solution}, given by 
\begin{align*}
V^\en&\bydef N_\J^* p(1-\lambda) =\frac{\mu p(1-\lambda)}{\mu+\nu}, \\ 
S^\en_\J&\bydef \frac{ N_J^*(1-p)\nu }{\nu + \beta I^\en} =\frac{\mu\nu(1-p)}{(\mu+\nu)(\nu+\beta I^\en)} ,\\
R^\en_J&\bydef \frac{\gamma}{\nu}I_\J^\en + N_\J^* p\lambda=N_J^*\left[\frac{(1-p)\gamma\beta I^\en}{(\gamma+\nu)(\nu+\beta I^\en)}+p\nu\right],\\
S^\en_\A&\bydef \mu N_\A^* \frac{(1-p)\nu+p(1-\lambda)(\nu+\beta I^\en)}{(\mu+\beta I^\en)(\nu+\beta I^\en)}\ ,\\
R_\A^\en&\bydef \frac{\gamma}{\mu}I^\en + N_\A^* p\lambda,\\
I^\en_\J&\bydef \frac{ N^*_\J(1-p)\beta I^\en\nu }{(\nu+\gamma)(\nu + \beta I^\en)},\\
 I^\en_\A&\bydef \frac{\mu N^*_\A\beta I^\en }{\mu+\gamma}\left[\frac{p(1-\lambda)}{\mu+\beta I^\en}+\frac{(1-p)\nu}{(\mu+\beta I^\en)(\nu+\beta I^\en)}+\frac{ (1-p)\nu }{(\nu+\gamma)(\nu + \beta I^\en)}\right].
\end{align*}
Finally, the total number of infectious individuals at the endemic equilibrium is given by
\begin{equation}\label{Ien}
I^\en\bydef I^\en_\J+I^\en_\A=\frac{b_1+\sqrt{b_1^2+4b_2b_0}}{2b_2}\ ,
\end{equation}
where
\begin{align*}
b_0&\bydef \mu\nu\left[\beta (\mu(\mu+\gamma+\nu)(1-p)+\nu(\nu+\gamma)(1-\lambda p))-(\gamma+\mu)(\gamma+\nu)(\mu+\nu)\right]\\
&(b_0>0 \Leftrightarrow \mathcal{R}_p>1)\ ,\\
b_1&\bydef \beta^2 \nu\mu((\gamma+\nu)(1-\lambda p)+\mu(1-p))-\beta(\gamma + \mu)(\gamma + \nu)(\mu + \nu)^2, \\
b_2&\bydef \beta^2 (\gamma + \mu)(\gamma + \nu)(\mu+\nu)\ .
\end{align*}
\end{itemize}
\end{theorem}

\begin{proof}
 
The disease-free solution is immediate after imposing $I^\df_\J=I^\df_\A=0$ in the stationary (i.e., $'=0$) solution of the System~(\ref{model:eq_V})--(\ref{model:eq_RA}). Following~\cite[Thm. 2]{Driessche2002}, we consider the compartments corresponding to infectious individuals to be $x=(I_\J, I_\A)$ and the remaining compartments corresponding to non-infectious classes $y=(V,S_\J,R_\J,S_\A,R_\A)$. We define the rate of appearance of new infections as $\mathcal{F}(x,y)=(\beta(I_\J+I_{\A})S_\J,\beta(I_\J+I_{\A})S_{\A}))$ and the remaining transition terms as $\mathcal{V}(x,y)=(\nu I_J+\gamma I_J,-\nu I_J+(\gamma+\mu) I_\A)$. Hence, System \eqref{model:eq_V} can be written as
\[
x'=\mathcal{F}(x,y)-\mathcal{V}(x,y),\quad y'=g(x,y)\ ,
\]
for an appropriate function $g$.
We define the matrices 
\[
F=\left[\frac{\partial \mathcal{F}_i}{\partial x_j} (x_0,y_0)\right]=\begin{bmatrix} \beta S^{\df}_\J & \beta S^{\df}_\J\\
\beta S^{\df}_\A & \beta S^{\df}_\A \end{bmatrix}
\]
and 
\[
V=\left[\frac{\partial \mathcal{V}_i}{\partial x_j} (x_0,y_0)\right]=\begin{bmatrix}
\nu+\gamma & 0\\ -\nu & \mu+\gamma
\end{bmatrix}\ ,
\]
where $(x_0,y_0)$ represents the disease free equilibrium. 
It's straightforward to verify conditions $(A_1)$ to $(A_5)$ of Theorem 2 in \cite{Driessche2002}, hence we conclude that the effective reproduction number $\mathcal{R}_p$ is given by the spectral radius of the next generation matrix 
\[
FV^{-1}=\frac{\beta}{(\gamma+\mu)(\gamma+\nu)} \begin{bmatrix}
(\gamma+\mu+\nu)S^{\df}_\J & (\gamma+\nu)S^{\df}_\J\\ (\gamma+\mu+\nu)S^{\df}_{\ensuremath{\mathcal{A}}} & (\gamma+\nu)S^{\df}_{\ensuremath{\mathcal{A}}}
\end{bmatrix}\ ,
\]
i.e.,
\[
\mathcal{R}_p\bydef\frac{\beta}{\gamma+\mu}\left[\frac{\mu+\nu+\gamma}{\nu+\gamma}S_J^{\df}+S_A^{\df}\right]\ .
\]
The stability  follows from~\cite{Driessche2002}, namely the disease-free equilibrium is locally asymptotically stable if ${\mathcal{R}_p<1}$, and unstable if $ \mathcal{R}_p  > 1$.
 For the computation of the endemic equilibrium, we follow the same techniques as before; in this case, however, the stationary solution implicitly depends on the value of $I^\en$, the solution of $\wp(I)=0$, where $\wp(I)\bydef -b_2I^2+b_1I+b_0$.

For the case $\mathcal{R}_p=1$, note that $b_1<0$ and, therefore, $I^\en=0$, i.e., both equilibria coincide.
\end{proof}

As an immediate consequence of Theorem~\ref{Equilibriastability}, we write 
\[
\mathcal{R}_p=\mathcal{R}_0\left[1-p\left(1-\frac{(1-\lambda)\nu(\nu+\gamma)}{\mu(\mu+\gamma+\nu)+\nu(\nu+\gamma)}\right)\right] ,
\]
with 
\begin{equation}\label{R0}
{\ensuremath{\mathcal{R}}}_0\bydef \left.{\ensuremath{\mathcal{R}}}_p \right|_{p=0}=\frac{\beta}{\gamma+\mu}\left[\left(1+\frac{\mu}{\nu+\gamma}\right) N_J^*+N_A^*\right]\ .
\end{equation}

Furthermore,
\begin{theorem}\label{Equilibriastability2}
Let $\Gamma= \{(V, S_\J, I_\J, R_\J, S_\A, I_\A, R_\A): S_\J\leq S_\J^\df,  S_\A\leq S_\A^\df,\linebreak V\leq V^\df, N_\A\leq N_\A^*\}$, and consider  the  model given by Eqs.~(\ref{model:eq_V})--(\ref{model:eq_RA}). Then
 \begin{itemize}
    \item 
If $\mathcal{R}_p \le1$ the only equilibrium solution of the System (\ref{model:eq_V})--(\ref{model:eq_RA}) is the disease-free solution, which is globally asymptotically stable in $\Gamma$.

\item If $\mathcal{R}_p>1$ the disease-free solution is unstable. The System (\ref{model:eq_V})--(\ref{model:eq_RA}) is uniformly persistent. 
\end{itemize}
\end{theorem}

\begin{proof}
    The set $\Gamma$ is positively invariant. Following the notation from the proof of Thm.~\ref{Equilibriastability} we define $f(x,y)=(F-V)x-\mathcal{F}+\mathcal{V}$. We have that $f(x,y)\geq 0$  with $f(x,y_0)=0$ in $\Gamma$, $F\geq0$, $V^{-1}\geq 0$ and $V^{-1}F$ is irreducible. Moreover, $(0,y)=(0,N^*_\J,0,N^*_\A,0)$ is a globally asymptotically stable (GAS) equilibrium of the system $y'=g(0,y)$. Assuming $\mathcal{R}_p\neq 1$, we conclude that the disease-free solution is GAS in $\Gamma$ for $\mathcal{R}_p<1$ and that, for $\mathcal{R}_p>1$,  the system is uniformly persistent. For $\mathcal{R}_p=1$, there is a set in which the Lyapunov function has zero derivative. In this set, the global asymptotically stability of the disease-free equilibrium is proved through a direct analysis of the model equations, while in the remaining domain, the proof follows from the previous argument. For further details, see~\cite{Shuaietal2013}, in particular Thm. 2.2 and the comments immediately below.
\end{proof}

\begin{remark}
Persistence in a dynamical system implies that its solution will not arbitrarily approach zero. Uniform persistence is a stronger condition, in which the modulus of the solution is always larger than a strictly positive number, cf.~\cite{Thieme}. Models based on ordinary differential equations, as is the case of the System (\ref{model:eq_V})--(\ref{model:eq_RA}), ignore stochastic effects, i.e., they implicitly assume infinite populations. However, populations are always finite, and in non-uniformly persistent models, the solution will arbitrarily approach the disease-free boundary (i.e., $I=0$), which is absorbing, and therefore, the disease will be eventually extinct due to stochastic fluctuations. On the other hand, if solutions are uniformly persistent, then there is a minimum population size such that disease will not be extinct due to natural fluctuations. See also~\cite{ChalubSouza} for an intermediate model, based on drift-diffusion partial differential equations of degenerated type, that models both the deterministic evolution and the stochastic evolution. Degenerated partial differential equations models pose serious mathematical difficulties and will not be analyzed in the current work. 
\end{remark}

Finally, it is straightforward to prove that 
\begin{proposition}\label{corollary:pc}
Consider
\begin{equation}
\lambda>\lambda_c\bydef 1-\frac{(\gamma+\mu)(\mu+\nu)}{\beta\nu} 
 \end{equation}
and $\mathcal{R}_0>1$. Then, there is a critical vaccination coverage 
 \begin{equation}
 p_{\mathrm{c}}\bydef\frac{\mu(\mu+\gamma+\nu)+\nu(\nu+\gamma)}{\mu(\mu+\gamma+\nu)+\lambda\nu(\gamma+\nu)}\left(1-\frac{1}{\mathcal{R}_0}\right)\in (0,1)
 \end{equation}
such that for any $p> p_{\mathrm{c}}$ the disease free solution is globally asymptotically stable in $\Gamma$. 
\end{proposition}

\subsection{Social cost}
\label{ssec:social}

At the endemic equilibrium, we  define a social cost function (per unit of time) depending on the disease incidence and disease cost for both juveniles and adults and on the vaccination costs:

\begin{align}
\label{eq:phi}
\phi(p,\lambda)
&\bydef c_\A^{\mathrm{d}}(\beta(I^\en_\J+I^\en_\A)S^\en_\A+\nu I^\en_\J)+ c_\J^{\mathrm{d}}\beta(I^\en_\J+I^\en_\A)S^\en_\J+ c^\vac\delta\mu p N_\A^*\\
\nonumber
&=c_\A^{\mathrm{d}} (\gamma+\mu)I_\A^{\en}+c_\J^{\mathrm{d}} (\gamma+\nu) I_\J^{\en}+c^\vac\delta\mu p N_\A^*\\
\nonumber
 &=c_\A^{\mathrm{d}} [ (\gamma+\mu)I_\A^{\en}+\varepsilon (\gamma+\nu) I_\J^{\en}+r \delta\mu p N_\A^*] ,
\end{align}
where  $c_\A^{\mathrm{d}}>0$ and $c_\J^{\mathrm{d}}>0$ are the disease costs for adults and juveniles, respectively, and $c^\vac > 0$ is the vaccination cost. We define the relative costs $\varepsilon=c_\J^{\mathrm{d}}/c_\A^{\mathrm{d}}$ and $r=c^\vac/c_\A^{\mathrm{d}}$. Upon normalization, we will assume from now on that $c_\A^{\mathrm{d}}=1$. The fraction of the vaccination cost supported by the society is given by $\delta\in[0,1]$, where $\delta=1$ means that all cost is supported by the society (normally, the State), and $\delta=0$ means that the vaccinated individual pays the entire cost of the vaccination.

Note that $I^\en_{\A,\J}$ depend explicitly on $p$ and $\lambda$, cf.~Thm.~\ref{Equilibriastability}.

\begin{remark}\label{social_cost_remark}
Social cost $\phi(p,\lambda)$ was obtained at the endemic equilibrium and therefore it assumes equality between entry and exit rates in all compartments. While in the first line of Eq.~(\ref{eq:phi}), social costs are obtained from the entry rates, in the last line they follow from the exit rates; if one is interested in the total cost should multiply each term of the last expression by the average time of a given individual in each class, i.e., $1/(\gamma+\mu)$, $1/(\gamma+\nu)$, and $1/\mu$, representing exits by recovery or death, by recovery or aging, or by death, respectively.
\end{remark}

\begin{table}
\begin{center}
\begin{tabular}{cl}
\textbf{Parameter} &\textbf{Description}  \\
\hline\hline
$c_\A^{\mathrm{d}}$  & disease cost of an adult\\
   \hline
    $c_\J^{\mathrm{d}}$ & disease cost of a juvenile\\
     \hline
     $c^\vac $ &  vaccination cost\\
      \hline
     $\delta$ & Fraction of the vaccination costs supported by the society\\
      \hline
  $\varepsilon\bydef c_\J^{\mathrm{d}}/c_\A^{\mathrm{d}}$ & relative disease cost of juveniles vs. adults   \\
   \hline
   $r\bydef c^\vac/c_\A^{\mathrm{d}}$ &  relative vaccination cost vs. adults disease cost \\
\end{tabular}
\caption{Cost variables used in the model. Upon normalization $c_\A^{\mathrm{d}}=1$, results presented in this article will depend only on $\delta$, a modeling parameter, $\varepsilon$ and $r$. The values for the relative costs $\varepsilon$ and $r$  used in this work are arbitrary and used for illustration purposes.}\label{tab:costparam}%
\end{center}
\end{table}

We define the acceptable social-cost region as
\[
\mathcal{V} =\left\{ (p,\lambda) \in [0,1]\times [0,1]: \Phi_{\varepsilon,r,\delta}(p,\lambda)\bydef\phi(p,\lambda)-\phi_0\le 0\right\}\ ,
\]
where $\phi_0=\phi(0,0)$ is the social cost of the disease in an unvaccinated population. 

 We define two critical values: $\lambda_{\sup}$, below which social-cost acceptance depends on vaccine coverage $p$; and $\lambda_{\inf}$, below which social-cost is unacceptable for any vaccine coverage $p$. 
 \[
 \lambda_{\sup}=\sup_{\Phi(p,\lambda) >0} \lambda, \quad \lambda_{\inf}=\inf_{\Phi(p,\lambda) <0} \lambda\ .
 \]

 \begin{figure}
\centering
\begin{subfigure}{.5\textwidth}
  \centering
  \includegraphics[width=\linewidth]{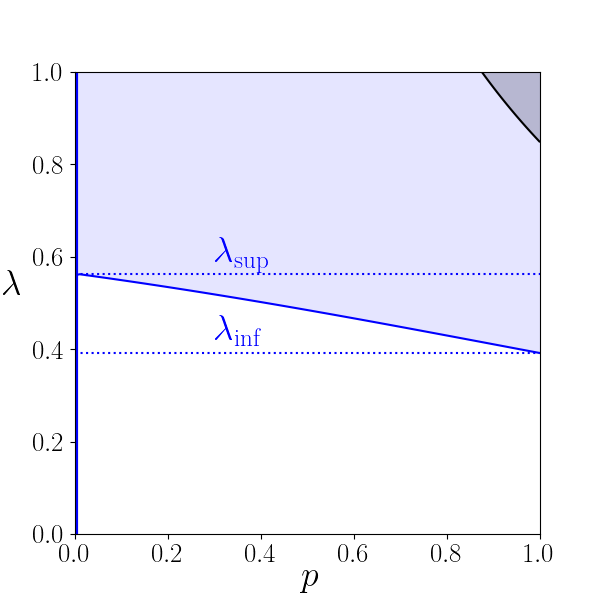}
\end{subfigure}%
\begin{subfigure}{.5\textwidth}
  \centering
  \includegraphics[width=\linewidth]{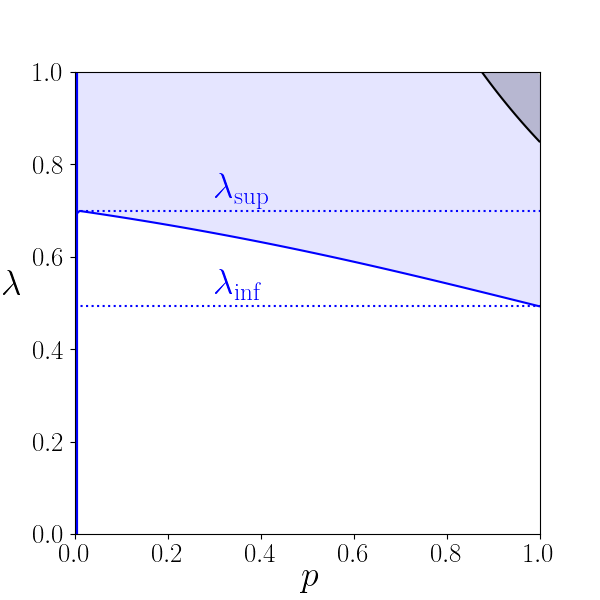}
\end{subfigure}
\caption{The  \colorbox{lightblue}{light-blue region} indicates the acceptable cost region $\Phi<0$, while the \colorbox{gray_lightblue}{grey region} is the disease-free region $\R_p<1$. We assume a juvenile/adult relative cost  $\varepsilon=0.15$, and a vaccine/disease cost $r=0.1$. Left: all vaccination costs are supported by the vaccinated individual, i.e., $\delta=0$. Right: the vaccination costs are supported by society with $\delta=1$.
}\label{Fig2}

\end{figure}
 
 Fig.~\ref{Fig2} illustrates the acceptable social-cost region in the parameter space $(p,\lambda)$  when $\delta=0$ and $\delta=1$. We observe that when $\delta$ increases the region gets smaller, due to the increase in the social cost, associated with the vaccination. Note that there is a subregion in which it is possible to eliminate the disease, i.e.,  ${\ensuremath{\mathcal{R}}}_p<1$.

\subsection{Individual cost and Nash equilibria}

Following \cite{Bauch04}, we assume that individuals freely choose to be vaccinated according to the perceived relative costs of the disease and of the vaccination.
For each $(p,\lambda)$, let us define $\Pi^{\nv}_{\A}$ and $\Pi^{\nv}_{\J}$  as the stationary (i.e., at equilibrium) probabilities of getting the disease as an adult and as a juvenile for unvaccinated individuals; and $\Pi^{\vac}_{\A}$ to be the stationary probability of getting the disease as an adult if vaccinated at birth. These values are equal to zero at the disease-free equilibrium and non-zero at the endemic equilibrium. Furthermore, they are continuous functions from the model parameters, cf.~\cite[Sec.~3.4]{Viana}.

We obtain explicit expressions for each of these three parameters. In the formulas below $*=\df$ or $*=\en$ according to $\mathcal{R}_p\le 1$ or $\mathcal{R}_p>1$, respectively, i.e., $*$ indicates the only stable equilibrium of the dynamics.

For $\Pi_\J^{\nv}$, we consider a given individual in the class $S_\J$, from which there are two possible exits. Either that given individual contracts the disease (and moves to the class $I_\J$) or he or she turns into an adult without being infected and moves to the class $S_\A$.  Explicitly,
\[
\Pi_\J^{\nv}(p,\lambda)\bydef\frac{\beta I^{\eq} S^{\eq}_\J}{(\beta I^{\eq}+\nu) S^{\eq}_\J}=\frac{\beta I^{\eq} }{\beta I^{\eq}+\nu}\ .
\]
The probability that a non-vaccinated adult gets the disease is given by the probability that a previously non-vaccinated juvenile does not get the disease as a juvenile times the probability that he or she gets the disease as an adult. Therefore
\[
\Pi_\A^{\nv}\bydef\left(1-\Pi_\J^{\nv}\right)\frac{\beta I^{\eq} S^{\eq}_\A}{(\beta I^{\eq}+\mu)S^{\eq}_\A}=\frac{\nu }{\beta I^{\eq}+\nu}\frac{\beta I^{\eq} }{\beta I^{\eq}+\mu}\ ,
\]
with $I^*\bydef I^*_\J+I^*_\A$.
Finally, the probability that a vaccinated adult gets the disease is the probability that the vaccine is effective only during the juvenile phase, $1-\lambda$, times the probability to get the disease from the class $S_\A$, i.e.,
\[
\Pi_\A^{\vac}\bydef(1-\lambda)\frac{\beta I^{\eq} }{\beta I^{\eq}+\mu}\ .
\]

We define the \emph{individual cost function} at endemic equilibrium, which corresponds to the expected cost of the individual strategy of being vaccinating with probability $q$ in a population with coverage $p$:
	\begin{align*}
			\Psi_{\varepsilon, r, \delta}(q,p,\lambda)&\bydef
			(1-q)(\Pi^{\nv}_{\A}+\varepsilon\Pi^{\nv}_\J)+q (\Pi^\vac_\A +r(1-\delta))\\
        &=\Pi^{\nv}_{\A}+\varepsilon \Pi^{\nv}_\J+q \left[ r(1-\delta)-\pi(p,\lambda)\right]\ ,
		\end{align*}
where the \emph{vaccination-infection risk index}, introduced in~\cite{Martins2017}, is given by
  \[
  \pi(p,\lambda)\bydef\Pi^{\nv}_{\A}(p,\lambda)+\varepsilon \Pi^{\nv}_\J(p,\lambda)-\Pi^\vac_\A(p,\lambda)\ .
  \]

The individual \emph{vaccination marginal expected payoff gain} $E(q,p)$ of an individual that uses the strategy of vaccinating with probability $q$ in a population that vaccinates with probability $p$ is given by 
\[
E(q,p)\bydef E(q,p; \varepsilon, r, \delta, \lambda)\bydef\Psi_{\varepsilon, r, \delta}(0,p,\lambda)-\Psi_{\varepsilon, r, \delta}(q,p,\lambda)\ .
\]
\begin{definition}\label{def:NashEquilibrium}
The population vaccination strategy $p_*$ is a vaccination \emph{Nash equilibrium}, if
\begin{align*}
E(q, p_*) - E(p_*, p_*) = (p_*-q)\left[r(1-\delta)-\pi(p_*,\lambda)\right] \leq 0 ,
\end{align*}
for every strategy $q \in [0,1]$. 
\end{definition}

In simple words, we say that the system is at Nash equilibrium if the vaccination coverage~$p_*$ is such that for every individual that uses a strategy~$q$ the expected payoff is not larger than the one it would have if the strategy~$p_*$ were used. 

\begin{proposition}\label{Proposition: Nash equilibria}
The model given by Eqs.~(\ref{model:eq_SJ})--(\ref{model:eq_RA}) has at least one Nash equilibrium. 
\end{proposition}

\begin{proof}
If $\pi(0,\lambda)\leq r(1-\delta)$, then $p_*=0$ is a Nash equilibrium. If $\pi(1,\lambda)\geq r(1-\delta)$, then $p_*=1$ is a Nash equilibrium. If both inequalities are false there is at least one value of $p_*\in(0,1)$ such that $\pi(p_*, \lambda)=r(1-\delta)$ and $p_*$ is a Nash equilibrium.
\end{proof}

For high vaccine efficacy $\lambda > \lambda_{\mathrm{c}}$ and $\delta\in[0,1)$, the vaccination coverage that results from individuals' choices is below the elimination threshold~$p_c$, defined in Prop.~\ref{corollary:pc}.

\begin{proposition}\label{Prop:Bauch04_generalization}
    Let $\varepsilon,r>0$, $\delta \in [0,1)$, $\lambda\in [\lambda_{\mathrm{c}}, 1]$, where $\lambda_{\mathrm{c}}$ is given by Prop.~\ref{corollary:pc}. Let  $p_{\mathrm{c}}^\lambda$ given by Prop.~\ref{corollary:pc} and $p_*^\lambda$ a Nash equilibrium of the associated model. Then, 
 $p_*^\lambda<p_{\mathrm{c}}^\lambda$.    
\end{proposition}

\begin{proof}
    From Prop.~\ref{corollary:pc}, for any value $p>p_{\mathrm{c}}^\lambda$ it is true that $\pi(p,\lambda)=0$. From the continuity of~$\pi$, we conclude that $\pi(p_c^\lambda,\lambda)=0$.     
    Assume that $p_*^\lambda\ge p_{\mathrm c}^\lambda>0$. From Def.~\ref{def:NashEquilibrium} we have that $(p_*^\lambda-q)\left[r(1-\delta)-\pi(p_*^\lambda,\lambda)\right] \leq 0 $ for every $q \in [0,1]$, therefore $p_*^\lambda\leq q$, for every $q \in [0,1]$, which is a contradiction. 
\end{proof}

This result generalizes the idea, already present in~\cite{Bauch04}, that it is impossible to eradicate a disease through voluntary vaccination.

Inspired by the concept of evolutionary stable strategy in game dynamics, cf.~\cite{HofbauerSigmund}, we define:

\begin{definition}
The population vaccination strategy $p_*$ is an \emph{evolutionary stable vaccination (ESV)} strategy, if there is a $\tau_0 > 0$,
such that for every $\tau \in (0, \tau_0)$ and for every $q \in [0, 1]$, with $q \neq p_*$,
\[
E(q, (1-\tau)p_*+\tau q)- E(p_*, (1-\tau)p_*+\tau q) < 0 .
\]

\end{definition}

We are ready to state the conditions for the Nash equilibrium to be ESV.

\begin{proposition}\label{Proposition: Nash ESV}
Let $p_*$ be a Nash equilibrium of the vaccination game. If $p_*=0$ or $p_*=1$, then $p_*$ is an ESV. Furthermore, if $\pi(p_*,\lambda)=r(1-\delta)$, $p_*$ is an ESV if and only if $\pi(p,\lambda)$ is decreasing in $p$ at $p=p_*$. In particular, $p_*\in(0,1)$ is an ESV if and only if $\pi(p,\lambda)$ is decreasing in $p$ at $p=p_*$.
\end{proposition}
\begin{proof}
This proof follows ideas from~\cite{Martins2017}. Let $p_*=0$ ($p_*=1$) be a Nash equilibrium. From Def.~\ref{def:NashEquilibrium}, we conclude that $\pi(0,\lambda)\le r(1-\delta)$ ($\pi(1,\lambda)\ge r(1-\delta)$, respect.). Assume that a strict inequality is valid. Let $\tau_0$ be small enough such that for all $\tau<\tau_0$, it is valid that $\pi(\tau q,\lambda)< r(1-\delta)$ ($\pi(1-\tau(1-q),\lambda)>r(1-\delta)$, respect.). It is clear that $E(q,\tau q)-E(0,\tau q)=-q(r(1-\delta)-\pi(\tau q,\lambda))<0$ ($E(q,1-\tau(1-q))-E(1,1-\tau(1-q))=(1-q)(r(1-\delta)-\pi(1-\tau(1-q),\lambda))<0$, respect.), for all $q\ne p_*$.

For the second part, note that $\pi(p_*,\lambda)=r(1-\delta)$ implies that 
\begin{align*}
    E(q,(1-\tau)p_*+\tau q)-&E(p_*,(1-\tau)p_*+\tau q)\\&=-(q-p_*)(\pi(p_*,\lambda)-\pi((1-\tau)p_*+\tau q,\lambda))\ ,
\end{align*}
and therefore $p_*$ is an ESV if and only if $\pi$ is decreasing in the first argument at $p=p_*$. The final result follows from Def.~\ref{def:NashEquilibrium}.
\end{proof}

\begin{figure}
\centering
\centering
\begin{subfigure}{.5\textwidth}
  \centering
  \includegraphics[width=\linewidth]{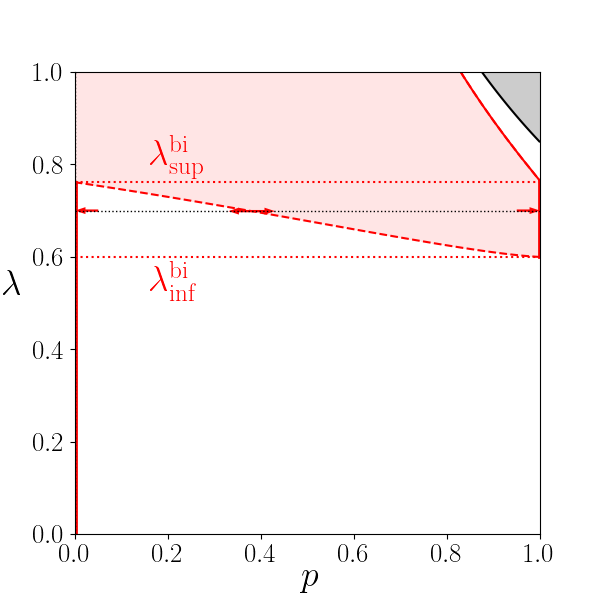}
\end{subfigure}%
\begin{subfigure}{.5\textwidth}
  \centering
  \includegraphics[width=\linewidth]{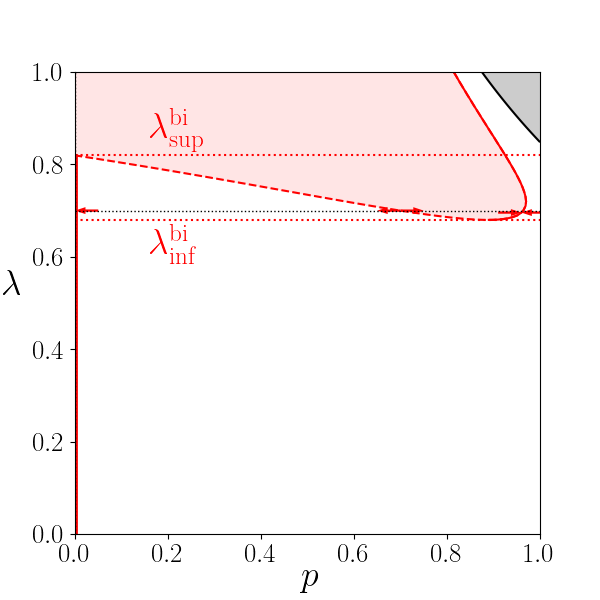}
\end{subfigure}
\caption{ Nash equilibria as a function of vaccine efficacy $\lambda$ and vaccination coverage $p$ for relative vaccination costs  $r=0.25$ (left) and $r=0.30$ (right). We assume a juvenile/adult relative cost $\varepsilon=0.15$ and all vaccination costs supported by the vaccinated individual ($\delta=0$). Red dashed and full lines correspond to unstable and stable Nash equilibria, respectively. The \colorbox{lightred}{light-red region} is the region for which the condition $\pi(p,\lambda)>r(1-\delta)$ is verified, cf. Prop.~\ref{Proposition: Nash ESV}. The horizontal black dotted line exemplifies the dynamics of rational individuals (indicated by the arrows) assuming a vaccine efficacy of $\lambda=0.70$. The \colorbox{lightgray} {grey region} is the disease-free region. 
}\label{Fig3}

\end{figure}

Fig.~\ref{Fig3} illustrates two possible situations described in Prop. \ref{Proposition: Nash ESV}: in the left, the two pure strategies are ESV and there exists an interior Nash equilibrium that is unstable; in the right, for higher relative vaccination costs there are two interior Nash equilibria, one unstable and the other stable, according to the monotonicity of $\pi$ with respect to $p$. The full (dashed) line indicates that $\pi$ is decreasing (increasing, respect.) with respect to $p$, with fixed $\lambda$. 
For both situations described, the region between $\lambda_{\inf}^{\mathrm{bi}}$ and $\lambda_{\sup}^{\mathrm{bi}}$ is the region for model bistability, in which we find three Nash equilibria, two stable (full red line) and one unstable (dashed red line) in between.

In the {light-red region} (that is, where the condition $\pi(p,\lambda)>r(1-\delta)$ is verified), a rational individual will accept to be vaccinated with a probability larger than the population average. In particular, in that region, there is an individual incentive to increase the overall vaccination coverage.

It is not possible to reach the disease-free region through voluntary vaccination if there is no incentive to be vaccinated, in agreement with Prop.~\ref{Prop:Bauch04_generalization}. Moreover, in case Fig.~\ref{Fig3}, left, the region in which there is no individual incentive to increase the vaccination coverage close to the disease-free region is disconnected from the set of vaccination coverage $p=0$.

\section{Numerical Examples}
\label{sec:numerics}

In this section, we present several numerical examples to discuss the present work. Parameters will be, except otherwise said, taken from Table~\ref{table:parameters}. In particular, varicella epidemiology fits our framework, as it is a mild disease for children that can have increased risk for adults and its inclusion in a universal vaccination program is debatable~\cite{Youngetal}. However, the framework developed here may be applied to several other situations, such as Zika, rubella, mumps and measles.

\begin{figure}
\centering
\includegraphics[width=.9\textwidth]{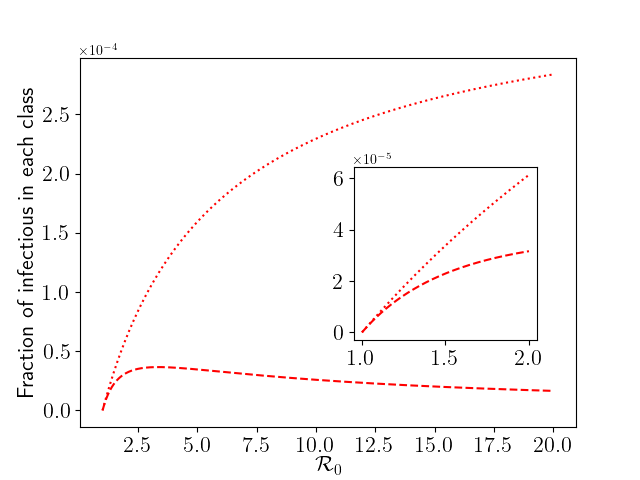}
\caption{Fraction of infectious juveniles (dotted) and adults (dashed), $I^\en_\J/N_\J^*$ and $I^\en_\A/N_\A^*$, respectively, as functions of $\mathcal{R}_0$. For $\mathcal{R}_0$ close to 1 both values are equal to first order, 
indicating that there is no essential difference between juveniles and adults in the disease dynamics (see box). After reaching a maximum around $\mathcal{R}_0\approx 3.35$ the fraction of infectious adults starts to decrease toward zero, while the fraction of infectious juveniles continues to grow, indicating that the disease is, in fact, a childhood disease.
}\label{Fig4}
\end{figure}

Fig.~\ref{Fig4} shows the fraction of infectious juvenile and adults at equilibrium as a function of the basic reproduction number without vaccination, i.e.,~$p=0$. The fraction of infectious juveniles is an increasing function of $\mathcal{R}_0$, but the fraction of infectious adults reaches a maximum around $\mathcal{R}_0\approx 3.35$. Furthermore, for values of $\mathcal{R}_0$ close to 1, the fraction of infectious is equal for both adults and juveniles, an indication that each individual is exposed to the pathogen approximately once during his or her life. This is indicated not only from the fact that both graphs are equal at $\mathcal{R}_0=1$, which is no surprise, as both values are equal to 0, but also from the fact that their slope is the same at this point. We conclude that a highly transmissible disease associated with permanent immunity will be, in equilibrium, a childhood disease. If the effect of this disease is mild in juveniles, there is no severe economic cost associated with the endemic state. This is the main reason we will always compare the economic cost associated with a vaccine program with vaccine efficacy $\lambda$ and vaccine coverage $p$ with the no-vaccination endemic state, cf. definition of $\Phi$ in Subsection~\ref{ssec:social}. For the parameters of Table~\ref{table:parameters}, most of the infectious individuals are below 15 years old, but a reasonable proportion of infectious individuals is above this value.

\begin{figure}
\centering
\centering
\begin{subfigure}{.5\textwidth}
  \centering
   \includegraphics[width=\linewidth]{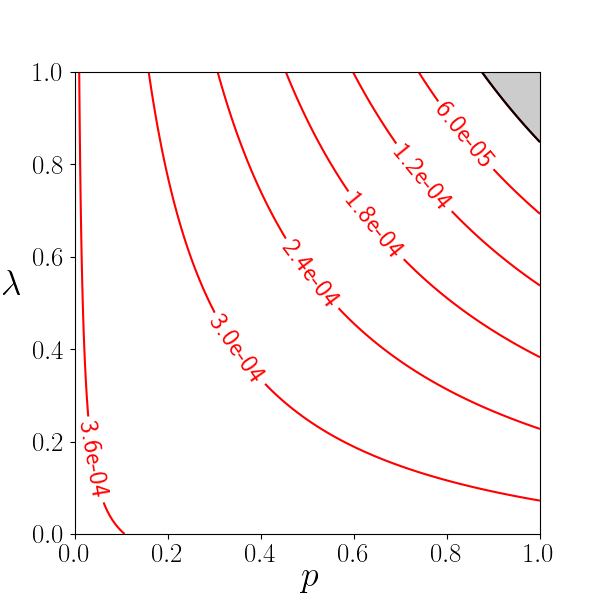}
\end{subfigure}%
\begin{subfigure}{.5\textwidth}
  \centering
  \includegraphics[width=\linewidth]{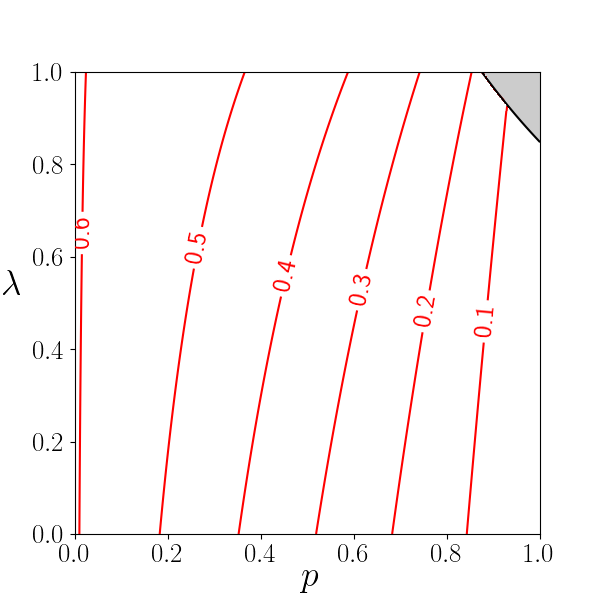}
\end{subfigure}
\caption{Left: Fraction of infectious individuals at endemic equilibrium assuming vaccination coverage $p$ and vaccine efficacy $\lambda$. Right: Fraction of infected juveniles, with respect to the number of infected individuals at endemic equilibrium, $I_\J^{\en}/I^{\en}$, as a function of the vaccine coverage $p$ and vaccine efficacy~$\lambda$. Note that increasing the vaccine coverage implies a smaller number of infected individuals but the disease became more relevant among adults. The \colorbox{gray_lightblue}{grey region} in the upper left corner of both graphs indicates the disease-free region.
}\label{Fig5}
\end{figure}

Assuming $\R_0=8$, the inclusion of a vaccination scheme is illustrated in Fig.~\ref{Fig5}. It clearly shows that for the considered set of parameters, the inclusion of a vaccination program will decrease the overall number of infectious individuals in the endemic equilibrium but it will increase the fraction of infections among adults. Therefore, the introduction of the vaccination scheme should be pondered to avoid negative outcomes for the population.

After the introduction of the vaccination, two natural questions arise: i) \emph{are people willing to be vaccinated?}, and ii) \emph{has the individual behavior a positive or negative effect on society?} The first question is addressed by introducing an individual cost of being vaccinated (that includes the perceived risk of the vaccine, eventual absence to work to go or to take the children to the vaccination site, the financial cost of buying the vaccine, etc) and the cost of non-being vaccinated, i.e., all the costs associated to contracting the disease. If the first is larger, then rational individuals will be vaccinated, if it is smaller, they will not. Interior Nash equilibria of the model correspond to points in which the equality holds. For the second question, we discuss if a given vaccination is in the acceptable social cost region. Ideally, we shall try to find a stable Nash equilibrium within that region, i.e., with $\Phi< 0$. However, this is not always possible.

\begin{figure}
\centering
\begin{subfigure}{.32\textwidth}
  \centering
  \includegraphics[width=\linewidth]{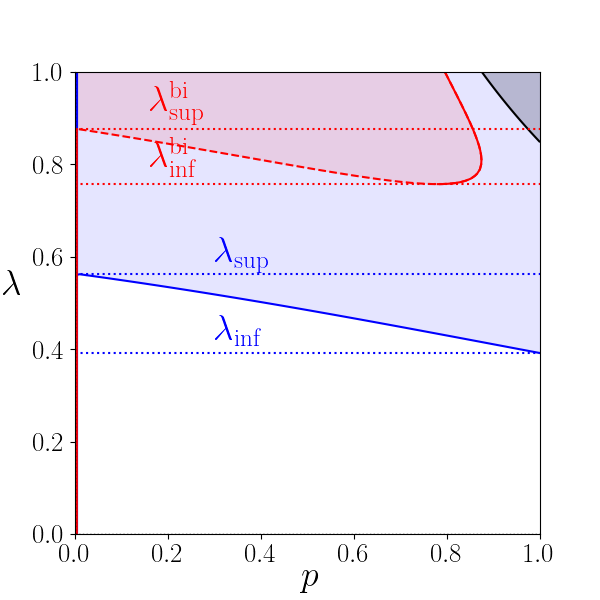}
  \caption{}
  \label{Fig6a}
\end{subfigure}%
\begin{subfigure}{.32\textwidth}
  \centering
  \includegraphics[width=\linewidth]{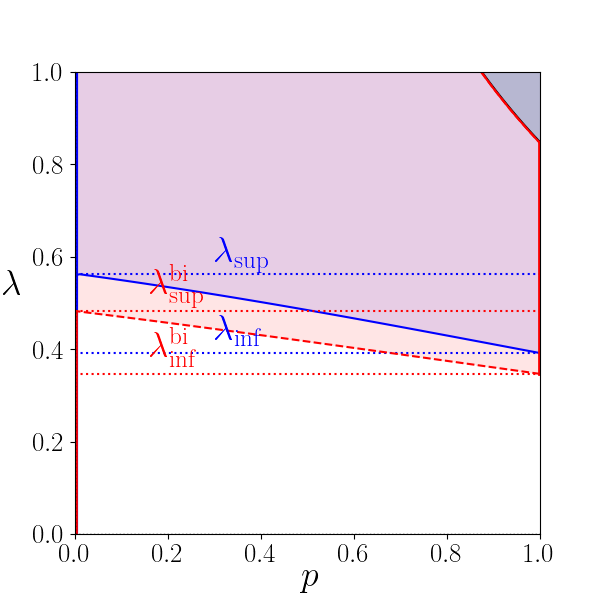}
  \caption{}
  \label{Fig6b}
\end{subfigure}
\begin{subfigure}{.32\textwidth}
  \centering
  \includegraphics[width=\linewidth]{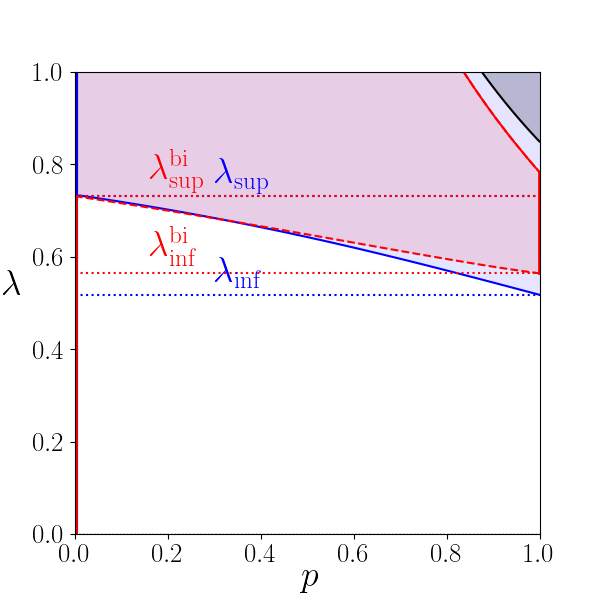}
  \caption{}
  \label{Fig6c}
\end{subfigure}
\caption{
Social \emph{vs.} individual interest with vaccine coverage $p$ and vaccine efficacy $\lambda$, with fixed $\varepsilon=0.15$. 
(a) The vaccination costs assumed by the individual, i.e., $\delta=0$, with a high-cost vaccine $r=0.35$;  (b) $\delta=0$ and low-cost vaccine $r=0.01$; (c) with shared costs between the individual and the society $\delta=0.36$ for the high-cost vaccine $r=0.35$; in the last case all Nash equilibria $p>0$ are in the socially cost-efficient region. The blue continuous line indicates the level set $\Phi=0$, and the red line corresponds to the set of Nash equilibria. Above the blue line, the social cost is negative, otherwise, it is positive.
The \colorbox{lightpurple}{light-purple region} is the intersection between \colorbox{lightblue}{light-blue} and \colorbox{lightred}{light-red regions}; see Figs.~\ref{Fig2} and~\ref{Fig3} for further explanation in the color code. 
}
\label{Fig6}
\end{figure}

Figs.~\ref{Fig6a} and~\ref{Fig6b} illustrate the regions on the parameter space $(p, \lambda)$ in which the individual and the social interests coincide and differ when all the individual vaccination costs are assumed by the beneficiary (i.e., there are no shared costs, as, for example, government subsidies). Fig.~\ref{Fig6c} introduces shared costs for high-cost vaccines, i.e., the society absorbs part of the individual cost.

Close to the disease-free region $\R_p<1$, there is always a barrier where there is no individual interest in being vaccinated, as the infection rates at that region will be residual. In Fig.~\ref{Fig6a}, the region where there is a social interest in increasing the vaccination, but there is an individual rejection of it, is a connected set. In Fig.~\ref{Fig6c} this region is disconnected. 

In Fig.~\ref{Fig6} we indicate the values of $\lambda_{\inf}$ ($\lambda_{\sup}$), the minimum efficacy such that for $p$ large enough (for all $p$, respect.) there is social interest in expanding the vaccine coverage and $\lambda_{\inf,\sup}^{\mathrm{bi}}$ the minimum and maximum values to the existence of bi-stable Nash vaccination equilibrium. Depending on the costs of the vaccine for society and individuals, their interests may not always agree: for a certain range of vaccine efficacy, it may be favorable for society to increase vaccination coverage, but due to the high cost of the vaccine, individuals choose not to be vaccinated, cf. Fig.~\ref{Fig6a} for $\lambda \in (\lambda_{\inf}, \lambda_{\inf}^{\mathrm{bi}})$. For a different set of parameters it may be favorable for individuals to vaccinate, due to the low vaccine cost, but not be beneficial for society, cf. Fig.~\ref{Fig6b} for $\lambda \in (\lambda_{\inf}^{\mathrm{bi}},\lambda_{\inf})$. This situation can be changed by allowing the vaccination costs to be shared. For example, in Fig.~\ref{Fig6c},  $\delta$ was chosen such that $\lambda_{\sup}^{\mathrm{bi}}=\lambda_{\sup}$; i.e., the Nash equilibrium at $p=0$ is socially neutral. In this case, individual vaccination is enhanced for lower vaccine efficacy, as compared to Fig.~\ref{Fig6a}. Moreover, in Fig.~\ref{Fig6c}, all Nash equilibrium $p>0$ are in the acceptable social cost region, avoiding individuals choosing to be vaccinated where their choice would increase social costs.

\begin{figure}
\centering
  \includegraphics[width=.5\textwidth]{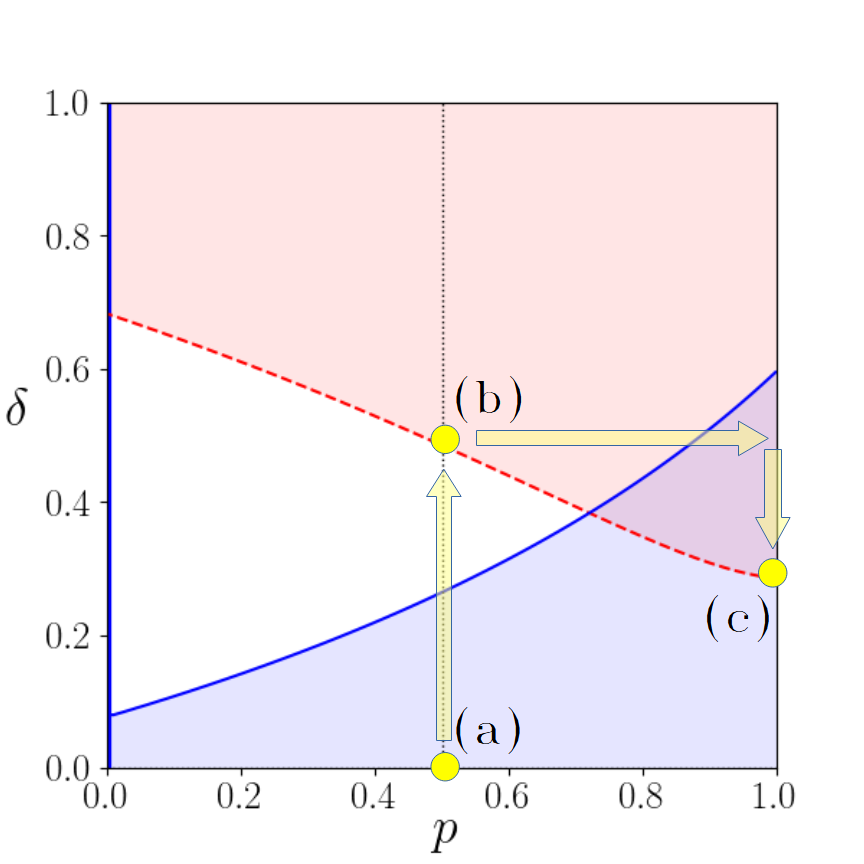}
\caption{
Effect of shared costs, $\delta$ being the fraction of the vaccination cost supported by the society, and $p$ being the vaccination coverage. We assume $\varepsilon=0.15$ and $\lambda=0.6$. In the regions below the blue line, the level of vaccination has a social cost lower than of no vaccination, while in the regions above the red line, a rational individual will choose to be vaccinated with a larger probability than the average individual. The \colorbox{lightpurple}{light-purple} region is the objective of the health authorities, where individuals freely decide to be vaccinated and the overall coverage is cost-efficient, i.e., $\Phi<0$. The yellow arrows stress a particular example of the sharing costs dynamic. The yellow points (a), (b), (c) are highlighted in Fig.~\ref{Fig8a}, ~\ref{Fig8b}, and ~\ref{Fig8c}, respectively. 
}\label{Fig7}
\end{figure}

\begin{figure}
\centering
\begin{subfigure}{.33\textwidth}
  \centering
  \includegraphics[width=\linewidth]{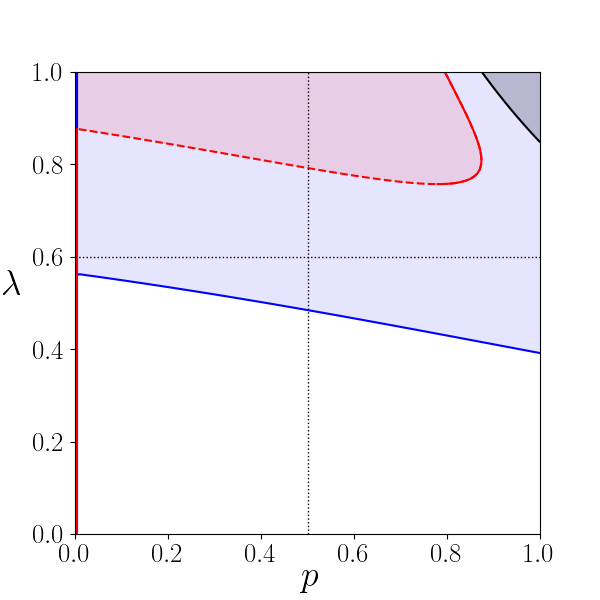}
  \caption{}
  \label{Fig8a}
\end{subfigure}%
\begin{subfigure}{.33\textwidth}
  \centering
  \includegraphics[width=\linewidth]{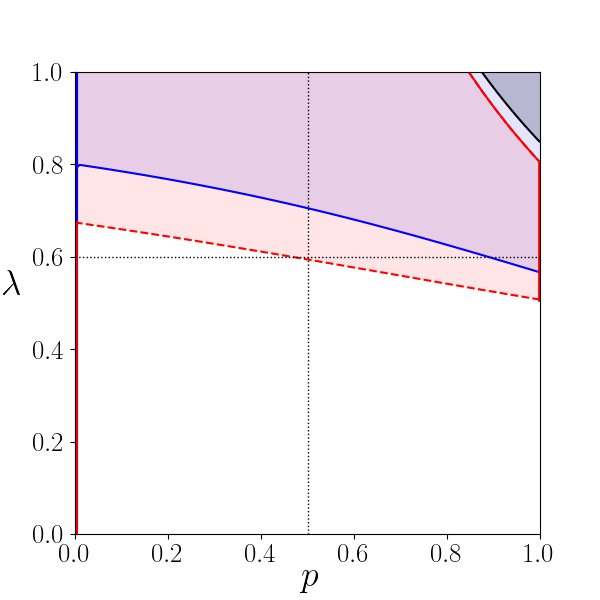}
  \caption{}
  \label{Fig8b}
\end{subfigure}
\begin{subfigure}{.33\textwidth}
  \centering
  \includegraphics[width=\linewidth]{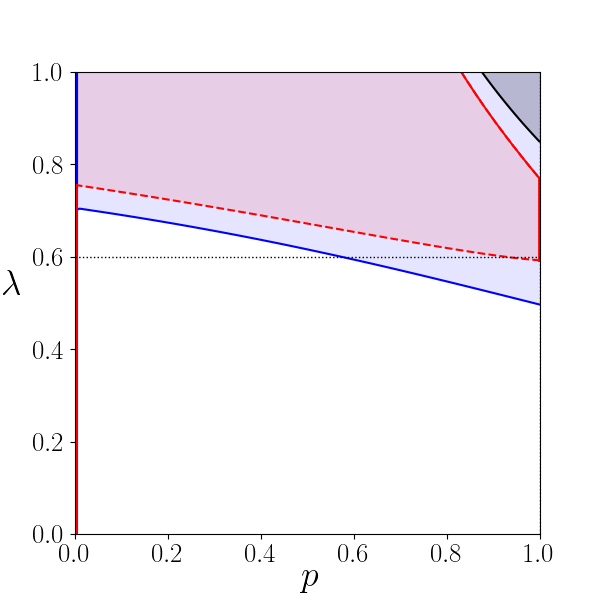}
  \caption{}
  \label{Fig8c}
\end{subfigure}
\caption{Parameter space $(p,\lambda)$ for different values of the fraction of the cost supported by the society:  (a) $\delta=0$, (b)  $\delta=0.5$, (c) $\delta=0.3$. 
The vertical black dotted line represents the level of vaccination $p=0.5$ in (a) and (b) and $p=1$ in~(c). The horizontal black dotted line represents the vaccine efficacy $\lambda=0.6$, corresponding to the examples studied in points (a), (b), and (c) in Fig.~\ref{Fig7}. 
}\label{Fig8}
\end{figure}

The effects of sharing costs are further explored in  Figs.~\ref{Fig7} and~\ref{Fig8}. Fig.~\ref{Fig7} shows a particular example, highlighted by the yellow arrows. Starting with a vaccination coverage of $p=0.5$ (indicated by a dotted vertical line) and $\delta=0$, rational individuals will not vaccinate, but vaccination will benefit society. Suppose vaccination cost starts to be shared between society and beneficiary (in this example, increasing the value of $\delta$ to above approximately 0.5). In that case, it is possible to create incentives such that rational individuals accept to be vaccinated, moving from the light-blue region to the light-red region, i.e. from point (a) to point (b). The natural dynamics will lead the population to a state in which the level of vaccination is high and the population is in a cost-efficient equilibrium, inside the acceptable socially-cost region. After that, it is possible to decrease vaccination incentives (in this case, the shared costs can be relaxed to approximately $\delta=0.3$) without decreasing vaccination coverage, i.e., moving towards point (c), a point in which both interests coincide with the largest share of the cost put in the individual. Fig.~\ref{Fig8a},~\ref{Fig8b}, and~\ref{Fig8c} show the superposition of the individual and social interests for different scenarios for shared costs corresponding to the three points depicted in Fig.~\ref{Fig7}, respectively.

\section{Discussion}
\label{sec:conc}

There is no doubt of the importance of vaccines to prevent and control several human diseases. They are cheap and safe when compared to the diseases they prevent. When a vaccination campaign is designed, it is important to understand not only the individual implications but the population-wide effects as well. The very concept of herd immunity is a population-wide effect of individual decisions. 

 In this work, we compared social \emph{vs}. individual interests regarding vaccination and disease costs and investigated if it is possible to promote voluntary vaccination and still satisfy both interests. For that, we considered an age-structured model with age-dependent costs, disease-induced permanent immunity, and imperfect vaccination and used a game theory approach to analyze individual decisions.

One of the interesting ideas that the use of game theory brought to the study of vaccinations (cf.~\cite{Bauch04}) is that it is not possible to reach the herd immunity level through voluntary vaccinations in a population of rational individuals. However, imperfect vaccination unfolds into qualitatively different situations, depicted in Fig.~\ref{Fig6}, that deserve to be discussed in more detail. Close to the herd immunity region (grey), there is always a region in which there is a social interest in increasing the vaccination coverage (light-blue region), but in which the risk perceived from the vaccine is larger than the risk from the disease. This is the region in which social and individual interests are at odds.  It can be connected or not to the region with no vaccination ($p=0$). This means that for reasons not modeled in the present work, but possibly associated with a slow variation in $\lambda$, for certain parameters it is possible to continuously move from a near-optimal situation to a region in which no individual is vaccinated. For lower vaccination costs supported by individuals (Fig.~\ref{Fig6b} and \ref{Fig6c}) no such thing happens.
The situation described in Fig.~\ref{Fig6c} is the desirable one, considering that disease elimination (grey region) is out of reach, and a light-purple region insulates a small light-blue region from the regions with low vaccine coverage. In the case of mumps and measles, most developed countries were close to disease elimination. However,  human behavior (particularly, due to the influence of vaccine deniers) made vaccine coverage decrease, which, possibly enhanced by immunity waning, moved these countries away from complete disease elimination. 

In the previous paragraph, we discussed the possible variation in time of the parameter $\lambda$. Several reasons may cause this variation, not only vaccine immunity waning. For example, for a disease such that it is recommended to have two doses to complete the immunization process, the parameter $\lambda$ may represent individuals with incomplete vaccination schemes, and, therefore, with partial and short-lived immunization~\cite{hong2021waning,shapiro2022effectiveness}.

The present work did not include any human dynamics. In particular, it does not consider changing in time of the vaccine coverage $p$. A possible simple dynamic is: $p$ increase if and when the new cases increase and decrease otherwise. A second possibility to an evolutionary law for $p$ is to consider the assumed total future risk of the disease --- with or without a discount due to future incertitude --- in the comparison between disease and vaccine risk made by rational individuals. See, e.g., ~\cite{Manfredi2013,Martins2017,Laguzet,Villota_2024} and references therein for modeling human behavior in epidemic models, coupling game-theoretical models and dynamical systems.

The model presented here is, as it is usual, an oversimplification of the
real-world effects. In particular, we consider that the disease dynamics is not age-structured, but its effect, more precisely, its cost, is age-dependent. As a consequence, we also simplified the effects of waning immunity, making it synchronous with the juvenile-adult transition (i.e., the cost-increase transition).  We do not consider directly derived infections (like shingles), or the elderly, immune-suppressed
individuals, babies, and several other categories that are socially relevant for disease dynamics, in particular, when social costs are considered. Of course, no specific disease fits precisely in this model, but the proposed framework can be tailored to accommodate the features of any particular disease dynamic.

\section*{Acknowledgements}
This work is funded by national funds through the FCT – Fundação para a Ciên\-cia e a Tecnologia, I.P.,\, under the scope of the projects UIDB/00297/2020  (https://doi.org/10.54499/UIDB/00297/2020) and UIDP/00297/2020\\ (https://doi.org/10.54499/UIDP/00297/2020) (Center for Mathematics and Applications\! --\! NOVA Math)\, and\, 2022.03091.PTDC\,\\ (https://doi.org/10.54499/2022.03091.PTDC), \emph{Mathematical Modelling of\\ Multiscale Con\-trol Systems: applications to human diseases} (CoSysM3). We acknowledge the insightful comments by Max Souza, Universidade Fluminense, Brazil and NOVA Math, Portugal and the three referees that helped to improve our manuscript. All the authors contributed equally to the development of the work's ideas, computational codes, data analysis, discussions, and writing of the final version of the manuscript.

\bibliographystyle{abbrv}
\bibliography{biblio}

\begin{thebibliography}{10}

\bibitem{Abushouk}
A.~I. Abushouk, A.~Negida, and H.~Ahmed.
\newblock {An updated review of Zika virus}.
\newblock {\em J. Clin. Virol.}, 84:53--58, NOV 2016.

\bibitem{Bauch04}
C.~T. Bauch and D.~J.~D. Earn.
\newblock {Vaccination and the theory of games}.
\newblock {\em {PNAS}}, {101}({36}):{13391--13394}, {2004}.

\bibitem{Brisson_2002}
M.~Brisson, N.~Gay, W.~Edmunds, and N.~Andrews.
\newblock Exposure to varicella boosts immunity to herpes-zoster: implications
  for mass vaccination against chickenpox.
\newblock {\em Vaccine}, 20(19):2500--2507, 2002.

\bibitem{ChalubSouza}
F.~A. C.~C. Chalub and M.~O. Souza.
\newblock Discrete and continuous sis epidemic models: A unifying approach.
\newblock {\em Ecol. Complex.}, 18(SI):83--95, JUN 2014.

\bibitem{Chang2020}
S.~L. Chang, M.~Piraveenan, P.~Pattison, and M.~Prokopenko.
\newblock {Game theoretic modelling of infectious disease dynamics and
  intervention methods: a review}.
\newblock {\em J. Biol. Dynam.}, 14(1):57--89, jan 2020.

\bibitem{Cohen2007}
C.~Cohen, J.~M. White, E.~J. Savage, J.~R. Glynn, Y.~Choi, N.~Andrews,
  D.~Brown, and M.~E. Ramsay.
\newblock Vaccine effectiveness estimates, 2004-2005 mumps outbreak, england.
\newblock {\em Emerg. Infect. Dis.}, 13(1):12--17, JAN 2007.

\bibitem{Daley}
A.~J. Daley, S.~Thorpe, and S.~M. Garland.
\newblock Varicella and the pregnant woman: Prevention and management.
\newblock {\em Aust. NZ J. Obstetrics \& Gynaecology}, 48(1):26--33, FEB 2008.

\bibitem{Damron_2020}
F.~H. Damron, M.~Barbier, P.~Dubey, K.~M. Edwards, X.-X. Gu, N.~P. Klein,
  K.~Lu, K.~H.~G. Mills, M.~F. Pasetti, R.~C. Read, P.~Rohani, P.~Sebo, and
  E.~T. Harvill.
\newblock Overcoming waning immunity in pertussis vaccines: Workshop of the
  national institute of allergy and infectious diseases.
\newblock {\em J. Immunol.}, 205(4):877--882, AUG 15 2020.

\bibitem{Doutor2016}
P.~Doutor, P.~Rodrigues, M.~C. Soares, and F.~A. C.~C. Chalub.
\newblock Optimal vaccination strategies and rational behaviour in seasonal
  epidemics.
\newblock {\em J. Math. Biol.}, 73:1437--1465, 2016.

\bibitem{Fine}
P.~E.~M. Fine.
\newblock Herd-immunity - history, theory, practice.
\newblock {\em Epidemiol. Rev.}, 15(2):265--302, 1993.

\bibitem{HofbauerSigmund}
J.~Hofbauer and K.~Sigmund.
\newblock {\em Evolutionary Games and Population Dynamics}.
\newblock Cambridge University Press, 1998.

\bibitem{hong2021waning}
K.~Hong, S.~Sohn, Y.~J. Choe, K.~Rhie, J.~K. Lee, M.~S. Han, B.~C. Chun, and
  E.~H. Choi.
\newblock Waning effectiveness of one-dose universal varicella vaccination in
  korea, 2011-2018: a propensity score matched national population cohort.
\newblock {\em J. Korean Med. Sci.}, 36(36), SEP 13 2021.

\bibitem{Kawai_BMC2014}
K.~Kawai, B.~G. Gebremeskel, and C.~J. Acosta.
\newblock Systematic review of incidence and complications of herpes zoster:
  towards a global perspective.
\newblock {\em BMJ Open}, 4(6), 2014.

\bibitem{Laguzet}
L.~Laguzet.
\newblock High order variational numerical schemes with application to
  {N}ash-{MFG} vaccination games.
\newblock {\em Ric. Mat.}, 67(1):247--269, 2018.

\bibitem{Youngetal}
Y.~H. Lee, Y.~J. Choe, J.~Lee, E.~Kim, J.~Y. Lee, K.~Hong, Y.~Yoon, and Y.-K.
  Kim.
\newblock Global varicella vaccination programs.
\newblock {\em Clin. Exp. Pediatr.}, 65(12):555--562, DEC 2022.

\bibitem{Lewnard2018}
J.~A. Lewnard and Y.~H. Grad.
\newblock {Vaccine waning and mumps re-emergence in the United States}.
\newblock {\em Sci. Transl. Med.}, 10(433), MAR 21 2018.

\bibitem{Manfredi2013}
P.~Manfredi and A.~D'Onofrio.
\newblock {\em Modeling the interplay between human behavior and the spread of
  infectious diseases}.
\newblock Springer New York, Oct. 2013.

\bibitem{Martins2017}
J.~Martins and A.~Pinto.
\newblock Bistability of evolutionary stable vaccination strategies in the
  reinfection {SIRI} model.
\newblock {\em B. Math. Biol.}, 79:853--883, 2017.

\bibitem{Mossong2003}
J.~Mossong and C.~P. Muller.
\newblock Modelling measles re-emergence as a result of waning of immunity in
  vaccinated populations.
\newblock {\em Vaccine}, 21(31):4597--4603, NOV 7 2003.

\bibitem{Mossong1999}
J.~Mossong, D.~J. Nokes, W.~J. Edmunds, M.~J. Cox, S.~Ratnam, and C.~P. Muller.
\newblock Modeling the impact of subclinical measles transmission in vaccinated
  populations with waning immunity.
\newblock {\em Am. J. Epidemiol.}, 150(11):1238--1249, DEC 1 1999.

\bibitem{Panagiotopoulos1999}
T.~Panagiotopoulos, I.~Antoniadou, E.~Valassi-Adam, and A.~Berger.
\newblock {Increase in congenital rubella occurrence after immunisation in
  Greece: retrospective survey and systematic review}.
\newblock {\em BMJ}, 319(7223):1462--1467, 1999.

\bibitem{Riera-Montes}
M.~Riera-Montes, K.~Bollaerts, U.~Heininger, N.~Hens, G.~Gabutti, A.~Gil,
  B.~Nozad, G.~Mirinaviciute, E.~Flem, A.~Souverain, T.~Verstraeten, and
  S.~Hartwig.
\newblock Estimation of the burden of varicella in europe before the
  introduction of universal childhood immunization.
\newblock {\em BMC Infect. Dis.}, 17, MAY 18 2017.

\bibitem{shapiro2022effectiveness}
E.~D. Shapiro and M.~Marin.
\newblock The effectiveness of varicella vaccine: 25 years of postlicensure
  experience in the {United States}.
\newblock {\em J. Infection.}, 226(SUPP 4, 4, SI):S425--S430, OCT 21 2022.

\bibitem{Shuaietal2013}
Z.~Shuai and P.~van~den Driessche.
\newblock Global stability of infectious disease models using lyapunov
  functions.
\newblock {\em SIAM J. Appl. Math.}, {73}({4}):{1513--1532}, {2013}.

\bibitem{Thieme}
H.~R. Thieme.
\newblock {\em Mathematics in population biology}.
\newblock Princeton Series in Theoretical and Computational Biology. Princeton
  University Press, Princeton, NJ, 2003.

\bibitem{Driessche2002}
P.~van~den Driessche and J.~Watmough.
\newblock {Reproduction numbers and sub-threshold endemic equilibria for
  compartmental models of disease transmission}.
\newblock {\em {Math. Biosci.}}, {180}({SI}):{29--48}, {2002}.

\bibitem{Viana}
M.~Viana and J.~M. Espinar.
\newblock {\em Differential equations---a dynamical systems approach to theory
  and practice}, volume 212 of {\em Graduate Studies in Mathematics}.
\newblock American Mathematical Society, Providence, RI, 2021.

\bibitem{Villota_2024}
J.~Villota-Miranda and R.~Rodriguez-Ibeas.
\newblock Simple economics of vaccination: public policies and incentives.
\newblock {\em Int. J. Health Econ. Ma.}, 2024 MAR 22 2024.

\bibitem{WANG20161}
Z.~Wang, C.~T. Bauch, S.~Bhattacharyya, A.~d'Onofrio, P.~Manfredi, M.~Perc,
  N.~Perra, M.~Salathé, and D.~Zhao.
\newblock Statistical physics of vaccination.
\newblock {\em Physics Reports}, 664:1 -- 113, 2016.

\bibitem{Yang_2020}
L.~Yang, B.~T. Grenfell, and M.~J. Mina.
\newblock Waning immunity and re-emergence of measles and mumps in the vaccine
  era.
\newblock {\em Curr. Opin. Virol.}, 40:48--54, FEB 2020.

\end{thebibliography}

\end{document}